\documentclass[11pt]{article}
\usepackage[margin=1in]{geometry}
\usepackage{amsmath,amssymb,amsthm,mathtools}
\usepackage{microtype}
\usepackage{enumitem}
\usepackage{tikz}
\usetikzlibrary{arrows.meta}
\usepackage[colorlinks=true,linkcolor=blue,citecolor=blue,urlcolor=blue]{hyperref}
\usepackage[authoryear]{natbib}
\newtheorem{lemma}{Lemma}
\newtheorem{theorem}{Theorem}

\theoremstyle{definition}
\newtheorem{result}{Result}
\theoremstyle{plain}
\newcommand{\E}{\mathbb{E}}
\newcommand{\D}{\mathcal{D}}
\newcommand{\OPT}{\mathrm{OPT}}
\newcommand{\pos}[1]{\left(#1\right)^+}

\title{A $1/2$-Approximation for Pandora's Problem with Deadlines}
\author{%
\begin{tabular}{cc}
\begin{minipage}[t]{0.4\textwidth}
    \centering
    Shenyu Qin\\
    \small Shanghai Jiao Tong University\\
    \small\href{mailto:qinshenyu29@gmail.com}{qinshenyu29@gmail.com}
\end{minipage}
&
\begin{minipage}[t]{0.4\textwidth}
    \centering
    Chihao Zhang\\
    \small Shanghai Jiao Tong University\\
    \small\href{mailto:chihao@sjtu.edu.cn}{chihao@sjtu.edu.cn}
\end{minipage}
\\[7ex]
\begin{minipage}[t]{0.4\textwidth}
    \centering
    Yuhao Zhang\\
    \small Shanghai Jiao Tong University\\
    \small\href{mailto:zhang_yuhao@sjtu.edu.cn}{zhang\_yuhao@sjtu.edu.cn}
\end{minipage}
&
\begin{minipage}[t]{0.4\textwidth}
    \centering
    Zixuan Zhu\\
    \small University of Hong Kong\\
    \small\href{mailto:zhuzixua@connect.hku.hk}{zhuzixua@connect.hku.hk}
\end{minipage}
\end{tabular}%
}
\date{}

\begin{document}
\maketitle

\begin{abstract}
Pandora's problem is a fundamental model of sequential search in which opening each box incurs a cost.
We study the variant with deadlines, where each box has a publicly known deadline for both opening and selection.
We prove that a simple threshold strategy achieves a $1/2$-approximation to the optimal expected utility, improving the previous guarantee arbitrarily close to $1/4$ due to Amanatidis et al. (2025).
\end{abstract}
\section{Introduction}

Pandora's problem was first introduced by \citet{Weitzman/1979/original} in the economics literature and models sequential search among costly alternatives.
It considers a set of Pandora's boxes.
Each box contains an uncertain value drawn from a known distribution, and opening a box reveals its value at a known cost. 
In each step, the decision-maker decides whether to open another box or to stop and take the maximum value from an opened box.
The goal is to maximize the value of the taken box minus the total inspection cost. 
In the classical setting, all boxes remain available throughout the search, and Weitzman's reservation-value rule gives an optimal strategy: it fixes an order for opening boxes in advance, but adapts the stopping decision to the values observed along the way. 
This model characterizes many real-life scenarios such as screening job candidates, comparing products, or conducting costly tests, where each inspection provides information but consumes resources.

In many selection problems, however, a decision-maker must pay to learn the value of an opportunity, and each opportunity remains available only for a limited time.
Continuing the search therefore entails both further evaluation costs and the risk of losing promising opportunities already identified.
For example, consider an employer interviewing candidates for a single position.
Each interview incurs a cost and reveals information about a candidate's suitability, but a candidate may accept another offer unless the employer makes a hiring decision by a specified date.
The employer may wish to interview more candidates before making an offer, yet continuing the search can cause a promising candidate it has already interviewed to become unavailable.
A deadline therefore limits both the opportunity to learn about an alternative and the time available to choose it.
Motivated by such settings, \citet{DBLP:conf/aaai/BergerEFF24} introduced Pandora's problem with deadlines, where each box $i$ has a publicly known deadline $d_i$.
Both opening and commitment must respect this deadline: box $i$ may be opened in round $t$ only if $t\leq d_i$, and the strategy may stop and commit to a previously opened box $i$ only by round $d_i$.
The goal remains to maximize the expected value of the selected box minus the total opening cost.

\citet{DBLP:conf/aaai/BergerEFF24} gave a partially adaptive, immediate-commitment algorithm with an approximation ratio of approximately $0.158$, later improved to arbitrarily close to $1/4$ by \citet{DBLP:journals/ai/AmanatidisBEFFRT25}.
Their common framework \citep{DBLP:conf/aaai/BergerEFF24,DBLP:journals/ai/AmanatidisBEFFRT25} connects the deadline model to a prophet inequality problem \citep{KrengelSucheston1977,KrengelSucheston1978} through a sequence of reductions, using reservation-value identities (no loss), an adaptivity-gap bound (factor $1/2$), and submodular maximization as \emph{black-box} components.
A PTAS exploiting the special structure of the submodular objective gives a factor of $1/(1+\delta)$ for any constant $\delta>0$.
It then invokes the classical prophet inequality (factor $1/2$).
Multiplying these factors yields $1/[4(1+\delta)]$, arbitrarily close to $1/4$.

We give a simple $1/2$-approximation to this problem by directly combining the reservation-value idea from Weitzman's rule with a \emph{white-box threshold--surplus analysis} of the classical prophet inequality.
We introduce the \emph{reservation-capped threshold surplus} function as a direct bridge between the deadline setting and the classical prophet inequality.
This connection avoids the intermediate adaptivity-gap and submodular-maximization losses, preserving the $1/2$ guarantee.
We explain these ideas and compare our approach with the previous framework in detail at the beginning of Section~\ref{sec:0.5}.

\begin{result}
\label{res:half-approximation}
We prove a $1/2$-approximation guarantee for Pandora's problem with deadlines, improving the previous guarantee arbitrarily close to $1/4$ (Theorem~\ref{thm:half-approximation}).
\end{result}

Beyond its approximation guarantee, our algorithm has two properties that are useful in practice: immediate commitment and partial adaptivity.

First, our algorithm uses \emph{immediate commitment}: upon opening a box, it either accepts its value and stops or rejects the box irrevocably.
This restriction is natural in applications where evaluating an opportunity must be followed by an immediate decision.
In the interview example, a hiring process may require the employer to make a hiring decision immediately after each interview.
Even when a delayed decision is allowed, the risk that a candidate accepts another offer may increase over time, making it difficult to specify a precise deadline for commitment.
Immediate commitment avoids relying on such a deadline and makes our algorithm applicable to these settings.
In our analysis, however, $\OPT$ may defer its commitment decision until stopping.
Deadlines constrain only its openings: it may select any previously opened box, even if that box's deadline has passed.
Our algorithm achieves the $1/2$ guarantee even against this more permissive benchmark.\footnote{The prior approximation guarantees \citep{DBLP:conf/aaai/BergerEFF24,DBLP:journals/ai/AmanatidisBEFFRT25} also hold for immediate-commitment algorithms against this benchmark.}
Moreover, Theorem~\ref{thm:immediate-commitment-tightness} shows that $1/2$ is the best possible approximation ratio for immediate-commitment algorithms against this benchmark.

Second, our algorithm is \emph{partially adaptive}: it fixes a set of boxes and their opening order in advance, but decides when to stop based on the observed values.
A \emph{fully adaptive} algorithm may also use these observations to choose which box to open next.
Partial adaptivity is useful when evaluations must be scheduled ahead of time.
For example, an employer may need to choose which candidates to interview and announce their interview order in advance.
Fixing this schedule avoids repeatedly reorganizing the interview process as new information arrives, while still allowing the employer to stop once it finds a suitable candidate.
Again, our analysis compares the algorithm with the fully adaptive optimum.
We also investigate the theoretical limits of partial adaptivity: Result~\ref{res:partial-adaptivity-limit} shows that this restriction can cause a strict loss relative to full adaptivity, even when all boxes have the same deadline.
In the classical Pandora's problem, Weitzman's rule is partially adaptive and achieves the fully adaptive optimum \citep{Weitzman/1979/original}.
Our result therefore establishes a separation between the two models: partial adaptivity suffices for optimality in the classical setting, but not in general under deadlines.

\begin{result}
\label{res:partial-adaptivity-limit}
We show that even with uniform deadlines, partially adaptive algorithms cannot guarantee an approximation ratio better than approximately $0.778$ relative to the fully adaptive optimum (Theorem~\ref{thm:partial-adaptivity-limit}).
\end{result}

\subsection{Other Related Work}

Beyond the related work discussed above, Pandora's problem has inspired a rich literature on a broad range of variants and extensions.

Most of the variations focus on opening or selection constraints, reward distributions, or the cost structure.
\citet{DBLP:conf/soda/Singla18} introduces the Price of Information framework, extending Pandora's problem to cost-minimization and utility-maximization settings with combinatorial selection or probing constraints.
Pandora's Problem with Nonobligatory Inspection has been studied in several lines of work \citep*[e.g.,][]{DBLP:journals/corr/abs-0804-1724,BeyhaghiK:EC:2019,DBLP:journals/jet/Doval18,ZivL:corr:2024}, where we can take a box without opening it and receive its expected value.
\citet{DBLP:conf/stoc/FuLL23} prove that the problem of finding the optimal strategy is NP-hard, and a PTAS is given in \citep{DBLP:conf/stoc/FuLL23,DBLP:conf/stoc/BeyhaghiC23}.
\citet{BoodaghiansFLL:Math.Oper:2023} impose precedence constraints on which boxes may be opened. 
Correlated Pandora's problem allows the box values to be arbitrarily correlated \citep{ChawlaGTTZ:FOCS:2020,DBLP:conf/approx/0001GMT23,GergatsouliT:NeurIPS:2023,DBLP:conf/focs/BansalHZ25}.
\citet{BergerEFF:EC:2023} consider non-additive, set-dependent inspection costs.
\citet{aouad2020pandora} investigate sequential inspections, where a box can first be inspected partially at a lower cost and then fully opened at a higher cost.

There are also models that consider other forms of exploration-exploitation tradeoffs.
Committed Pandora's Problem requires an immediate accept-or-reject decision after opening each box.
\citet{FuLX:ICALP:2018} gave a PTAS, and \citet{DBLP:conf/sigecom/Segev021} an EPTAS, relative to the optimal fully adaptive policy subject to the same commitment rule.
\citet{DBLP:conf/wine/AmanatidisFRT24} introduce the Pandora's Box Problem Over Time framework.
They study a temporal model in which rewards may decay after opening and opening costs vary over time.
\citet{EsfandiariHLM:AAAI:2019} consider boxes that arrive online, with decisions about inspection and selection made as they appear. 
\citet{DBLP:conf/icml/GergatsouliT22} study online learning of search orders for Pandora's problem and Min Sum Set Cover, including settings with bandit feedback. 
In market design, \citet{DBLP:conf/sigecom/KleinbergWW16} show how a descending-price auction can coordinate buyers' costly investigations. For a broader overview of Pandora's problem, its variants, and applications, we refer the reader to a recent survey by \citet{DBLP:journals/sigecom/BeyhaghiC23}.
\section{Preliminaries}
\label{sec:pre}
We write $[n]$ for $\{1,2,...,n\}$ with the convention that $[0]$ is the empty set. 
Let $x^+$ denote the
function $\max\{x,0\}$.
\subsection{Pandora's Problem with Deadlines}
Consider a set of boxes $B=[n]$.
Each box $i$ has a nonnegative random value $v_i$, an opening cost $c_i\geq 0$, and an integer deadline $d_i\geq 1$.
The costs and deadlines are known in advance, and the values are drawn independently from known distributions $\mathcal{D}_i$ with finite expectations.
Opening box $i$ incurs cost $c_i$ and reveals $v_i$.

In round $t=1,2,\ldots$, a strategy may open one previously unopened box $i$, provided that $t\leq d_i$.
After the opening, it may continue to the next round or irrevocably commit to an available box and stop.
It may also stop without committing to any box. 
A strategy therefore commits to at most one previously opened box.

\paragraph{Commitment.}
We distinguish three rules governing which opened boxes remain available for commitment.
Opening deadlines apply under all three rules.
\begin{itemize}
    \item \emph{Immediate commitment.} Immediately after opening a box, the strategy must either commit to it and stop or discard it irrevocably.
    \item \emph{Deferred-to-deadline commitment.} An opened box remains available until its deadline. The strategy may commit to box $i$ in round $t$ only if $t\leq d_i$.
    \item \emph{Deferred-until-stopping commitment.} Every opened box remains available until the strategy stops. It may commit to any previously opened box, even if that box's deadline has passed.
\end{itemize}
These rules are increasingly permissive, with immediate commitment imposing the strongest restriction.

For a strategy $\pi$, let $O(\pi)$ be its random set of opened boxes and let $C(\pi)$ be the box to which it commits.
Set $C(\pi)=0$ if it makes no commitment, with $v_0=0$.
Its expected net utility is
\[
 u(\pi)=\E\!\left[v_{C(\pi)}-\sum_{i\in O(\pi)}c_i\right].
\]
In particular, the reward is the value of the committed box, and every opened box incurs its cost.

\paragraph{Adaptivity.}
We distinguish two levels of adaptivity in the opening and stopping decisions.
\begin{itemize}
    \item A \emph{fully adaptive} strategy may choose the next box to open and decide when to stop based on all values observed so far.
    \item A \emph{partially adaptive} strategy fixes an ordered list of boxes before observing any values and opens a prefix of this list. Its stopping time may depend on the observed values.
\end{itemize}
These definitions concern the opening sequence and stopping time; the choice of which box to commit to may depend on the observed values, subject to the commitment rule.

\paragraph{Algorithm and benchmark.}
Our algorithm in Section~\ref{sec:0.5} is partially adaptive and uses immediate commitment, the most restrictive commitment rule.
For its analysis, we use the most permissive benchmark: a fully adaptive strategy with deferred-until-stopping commitment.
We write
\[
 \OPT=\sup_{\pi}u(\pi),
\]
where the supremum is over all such strategies that respect the opening deadlines.
This benchmark may commit to an expired box, provided that the box was opened by its deadline.
A guarantee against $\OPT$ therefore also holds against the fully adaptive optimum under deferred-to-deadline or immediate commitment.

\subsection{Weitzman's Rule}
Under either deferred-to-deadline or deferred-until-stopping commitment, setting all deadlines to $+\infty$ yields the classical Pandora's problem, for which Weitzman's rule is optimal \citep{Weitzman/1979/original}.
For each box $i$, its reservation value $r_i$ is defined such that
\[
c_i = \mathbb{E}_{v_i\sim\mathcal{D}_i}\bigl[(v_i-r_i)^+\bigr].
\]
When $c_i=0$, we take $r_i=\operatorname{ess\,sup}v_i$, allowing $r_i=+\infty$.
The rule opens an unopened box with the largest reservation value and stops when the largest value observed so far is at least the largest reservation value among the unopened boxes. 
It then commits to the opened box with the highest realized value.
Its opening order is fixed by the reservation values, while its stopping time depends on the observed values, so the rule is partially adaptive.

\paragraph{Failure under deadlines.}

Weitzman's rule accounts for opening costs and potential rewards, but its reservation-value order does not account for deadlines.
With deadlines, opening one box can make another unavailable, so the opening order affects not only costs and information but also which options remain feasible.

We give an instance in which Weitzman's rule, even with deferred-until-stopping commitment, performs arbitrarily poorly relative to the optimal immediate-commitment strategy.
Suppose there are two boxes with deadlines $d_1=1$ and $d_2=2$. 
Their costs are $c_1=1-\frac{1}{N}$ and $c_2=1-\frac{1}{N}-\frac{1}{N^2}$, respectively, where $N$ is a large real number. 
The first box has value $v_1=N$ with probability $\frac{1}{N}$ and value $v_1=1$ otherwise. 
The second box has value $v_2=N$ with probability $\frac{1}{N}$ and value $v_2=0$ otherwise.

By definition, the reservation values are $r_1=1$ and $r_2=1+\frac{1}{N}$. 
Since $r_1<r_2$, Weitzman's rule opens the second box first. 
The first box can no longer be opened after round 1, so the rule can only commit to the second box, even under deferred-until-stopping commitment.
Its expected utility is
\[
N\cdot\frac{1}{N}
-\left(1-\frac{1}{N}-\frac{1}{N^2}\right)
=\frac{1}{N}+\frac{1}{N^2}.
\]
In contrast, opening the first box in round 1 and committing to it immediately gives expected utility
\[
N\cdot\frac{1}{N}
+1\cdot\left(1-\frac{1}{N}\right)
-\left(1-\frac{1}{N}\right)
=1.
\]
The optimal immediate-commitment strategy therefore has expected utility at least $1$. Thus, even with deferred-until-stopping commitment, Weitzman's rule achieves a ratio of at most $1/N+1/N^2$ against this optimum, which tends to zero as $N\to\infty$.

The example isolates a failure of the reservation-value opening order under deadlines.
Box 2's slightly larger reservation value reflects its rare high payoff and lower opening cost, making it attractive under Weitzman's rule.
Opening it first leaves no opportunity to open box 1 by its deadline; when box 2 yields zero, there is no fallback.
Deferred-until-stopping commitment preserves the option to choose among boxes already opened, but it cannot recover an unopened box whose deadline has passed.
Opening box 1 first and committing immediately avoids this loss.

\section{Our Algorithm}
\label{sec:0.5}
This section presents our main result: a $1/2$-approximate partially adaptive algorithm with immediate commitment.
Our approach builds a direct connection between the deadline setting and classical prophet inequalities.
The algorithm has two stages. First, it selects a feasible set of boxes and fixes an opening order that respects their deadlines.
Second, it applies an immediate-commitment policy to the selected boxes.
The algorithm is partially adaptive: its opening order is fixed in advance, while its stopping time depends on the observed values.
The guarantee holds against $\OPT$, the fully adaptive benchmark with deferred-until-stopping commitment.

The key to this two-stage framework is to determine what makes a feasible set $S$ good and how much of its value can be preserved under immediate commitment.
The prior algorithms \citep{DBLP:conf/aaai/BergerEFF24,DBLP:journals/ai/AmanatidisBEFFRT25} follow the same framework.
Their approach first uses the reservation-value property underlying Weitzman's rule as a \emph{black-box} tool to associate each feasible set $S$ with the benchmark $B(S)=\E[\max_{i\in S}\min\{v_i,r_i\}]$, where $r_i$ is box $i$'s reservation value.
It then reduces the immediate-commitment step on $S$ to a prophet inequality on the capped values and invokes the classical $1/2$ guarantee as a black box, obtaining expected net utility at least $B(S)/2$.
To find a feasible set whose benchmark is large relative to $\OPT$, their framework uses a sequence of reductions: a relaxation to stochastic probing, an adaptivity-gap bound relating the adaptive probing optimum to the best fixed feasible set, and submodular maximization to approximately optimize $B(S)$.
The adaptivity-gap bound guarantees that the best feasible benchmark is at least $\OPT/2$.
A PTAS exploiting the special structure of $B(S)$ contributes a further factor of $1/(1+\delta)$ for any constant $\delta>0$.
Together, these three factors give an approximation ratio of $1/[4(1+\delta)]$, arbitrarily close to $1/4$.

\paragraph{Reservation-capped threshold surplus: a simple, direct bridge.} We obtain a more direct connection through a \emph{white-box} analysis that combines Weitzman's reservation-value idea with the \emph{threshold--surplus decomposition} underlying the classical prophet inequality \cite{SamuelCahn1984}.
The resulting bridge is the \emph{reservation-capped threshold surplus}, defined by
\[
 w_i(t)=\E[(\min\{v_i,r_i\}-t)^+].
\]
These surpluses define the set benchmark $\mathcal{B}_t(S)=t+\sum_{i\in S}w_i(t)$.
For every fixed threshold, the maximum benchmark value over feasible sets directly upper-bounds $\OPT$.
At a balancing threshold, a simple threshold algorithm with immediate commitment achieves at least half of this value.
We design a polynomial-time dynamic program to perform the maximization exactly for each fixed threshold.
As illustrated in Figure~\ref{fig:framework-comparison}, this direct connection avoids the intermediate adaptivity-gap and submodular-maximization losses, preserving the $1/2$-approximation guarantee against the adaptive optimum.

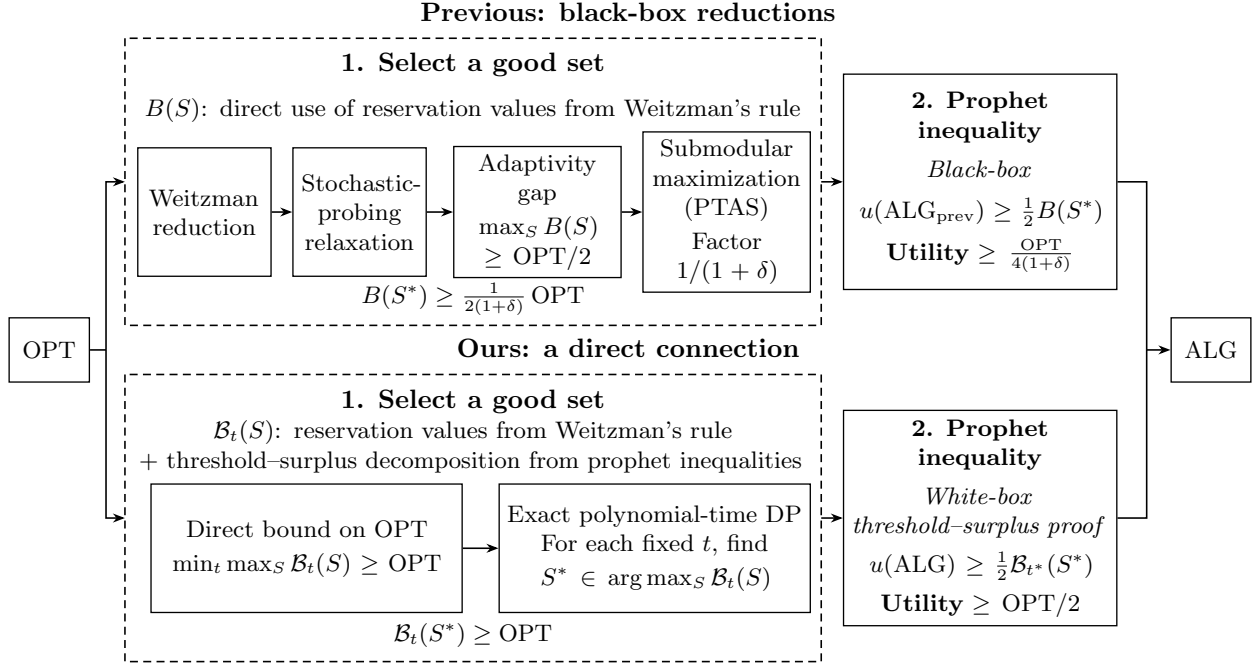
\begin{figure}[!t]
\centering
\begin{tikzpicture}[
    font=\footnotesize,
    stage/.style={draw, line width=0.5pt, align=center,
        minimum height=1.7cm, inner sep=3pt},
    endpoint/.style={stage, font=\footnotesize\bfseries,
        minimum width=1.05cm, minimum height=0.85cm,
        inner xsep=5pt, inner ysep=4pt},
    connector/.style={line width=0.6pt},
    flow/.style={connector, -{Stealth[length=1.6mm,width=1.2mm]}}
]
\node[endpoint] (start) at (0.45,-2.225) {$\OPT$};
\node[endpoint] (finish) at (15.8,-2.225) {$\mathrm{ALG}$};
\node[font=\small\bfseries] at (8.1,2.25)
    {Previous: black-box reductions};
\node[font=\small\bfseries] at (8.1,-2.20)
    {Ours: a direct connection};

\draw[densely dashed, line width=0.6pt]
    (1.45,1.9) rectangle (10.65,-1.9);
\draw[densely dashed, line width=0.6pt]
    (1.45,-2.55) rectangle (10.65,-6.35);
\node[font=\small\bfseries] at (6.05,1.55) {1. Select a good set};
\node[font=\small\bfseries] at (6.05,-2.90) {1. Select a good set};

\node[align=center, text width=8.7cm] at (6.05,0.95) {
    $B(S)$: direct use of reservation values from Weitzman's rule
};
\node[align=center, text width=8.7cm] at (6.05,-3.50) {
    $\mathcal{B}_t(S)$: reservation values from Weitzman's rule\\
    $+$ threshold--surplus decomposition from prophet inequalities
};

\node[stage, text width=1.55cm] (reservation) at (2.5,-0.40) {
    Weitzman\\reduction
};
\node[stage, text width=1.55cm] (probing) at (4.55,-0.40) {
    Stochastic-\\probing\\relaxation
};
\node[stage, text width=2cm] (adaptivity) at (6.9,-0.40) {
    Adaptivity gap\\[0.2em]
    $\max_S B(S)$\\$\geq\OPT/2$
};
\node[stage, text width=2cm] (submodular) at (9.4,-0.40) {
    Submodular\\maximization\\(PTAS)\\[0.2em]
    Factor $1/(1+\delta)$
};
\node at (6.05,-1.55) {$B(S^*)\geq\frac{1}{2(1+\delta)}\,\OPT$};

\node[stage, text width=3.9cm] (directbound) at (3.85,-4.85) {
    Direct bound on $\OPT$\\[0.2em]
    $\min_t\max_S\mathcal{B}_t(S)\geq\OPT$
};
\node[stage, text width=3.9cm] (selection) at (8.45,-4.85) {
    Exact polynomial-time DP\\
    For each fixed $t$, find\\[0.2em]
    $S^*\in\arg\max_S\mathcal{B}_t(S)$
};
\node at (6.05,-6.00) {$\mathcal{B}_t(S^*)\geq\OPT$};

\node[stage, text width=3.4cm, minimum height=2.85cm]
    (previouscommitment) at (12.75,0) {
    \textbf{2. Prophet inequality}\\[0.4em]
    \emph{Black-box}\\[0.5em]
    $u(\mathrm{ALG}_{\mathrm{prev}})\geq\tfrac12 B(S^*)$\\[0.5em]
    \textbf{Utility $\geq\frac{\OPT}{4(1+\delta)}$}
};
\node[stage, text width=3.4cm, minimum height=2.85cm]
    (ourcommitment) at (12.75,-4.45) {
    \textbf{2. Prophet inequality}\\[0.4em]
    \emph{White-box}\\
    \emph{threshold--surplus proof}\\[0.4em]
    $u(\mathrm{ALG})\geq\tfrac12\mathcal{B}_{t^*}(S^*)$\\[0.4em]
    \textbf{Utility $\geq\OPT/2$}
};

\draw[connector] (start.east) -- (1.2,-2.225);
\draw[flow] (1.2,-2.225) |- (1.45,0);
\draw[flow] (1.2,-2.225) |- (1.45,-4.45);
\draw[flow] (reservation.east) -- (probing.west);
\draw[flow] (probing.east) -- (adaptivity.west);
\draw[flow] (adaptivity.east) -- (submodular.west);
\draw[flow] (directbound.east) -- (selection.west);
\draw[flow] (10.65,0) -- (previouscommitment.west);
\draw[flow] (10.65,-4.45) -- (ourcommitment.west);
\draw[connector] (previouscommitment.east) -| (14.95,-2.225);
\draw[connector] (ourcommitment.east) -| (14.95,-2.225);
\draw[flow] (14.95,-2.225) -- (finish.west);
\end{tikzpicture}
\caption{Comparison of the common black-box framework \citep{DBLP:conf/aaai/BergerEFF24,DBLP:journals/ai/AmanatidisBEFFRT25} and our white-box approach. The PTAS exploits the objective's special structure.}
\label{fig:framework-comparison}
\end{figure}

\subsection{Reservation-capped Threshold Surplus}
We first define the reservation-capped threshold surplus and its associated benchmark for a feasible set of boxes. We show that maximizing this benchmark over feasible sets upper-bounds $\OPT$, and give a dynamic program that finds a maximizing set for each fixed threshold.

For each box $i$, let $y_i=\min\{v_i,r_i\}$ be its realized value capped at its reservation value $r_i$.
The excess above this cap is $o_i=v_i-y_i=(v_i-r_i)^+$. 
For a threshold $t\geq 0$, let $w_i(t)=\mathbb{E}[(y_i-t)^+]$ denote the \emph{reservation-capped threshold surplus} of box $i$, namely the expected amount by which its capped reward exceeds $t$.

A set is \emph{deadline-feasible} if its boxes can be opened legally, one per round, in some order.
Let $\mathcal{D}$ be the family of deadline-feasible sets, including the empty set.

For $S\in\D$ and $t\geq 0$, define
\[
\mathcal{B}_t(S)=t+\sum_{i\in S}w_i(t),
\qquad
W(t)=\max_{S\in\D}\sum_{i\in S}w_i(t).
\]
Thus, $\max_{S\in\D}\mathcal{B}_t(S)=t+W(t)$.

To relate this benchmark to $\OPT$, we use the reservation-value identity to bound net utility by the expected capped value of the selected box.
The selected capped value is at most $t$ plus the total surplus above $t$ of all opened boxes.
Since the decision to open each box is independent of its value and the opened set is always deadline-feasible, the expected total surplus is at most $W(t)$.

\begin{lemma}[Adaptive-surplus upper bound]
For every adaptive strategy $\pi$ that respects the opening deadlines and every $t\geq 0$,
\[
 u(\pi)\leq t+W(t)=\max_{S\in\D}\mathcal{B}_t(S).
\]
\end{lemma}

\begin{proof}
Fix a threshold $t\geq 0$ and an adaptive strategy $\pi$ that respects the opening deadlines.
Let $I_i$ indicate whether $\pi$ opens box $i$.
Let $A_i$ indicate whether $\pi$ commits to box $i$ at stopping.
Since $\pi$ commits to at most one opened box, we have $0\leq A_i\leq I_i$ and $\sum_i A_i\leq 1$.

Since
$y_i \leq t + \pos{y_i-t}$, we have
\begin{equation}
    \label{eq:decomposite_y}
 \sum_i A_i \cdot y_i\leq t+\sum_i I_i \cdot \pos{y_i-t}.
\end{equation}

Using $y_i\leq t+(y_i-t)^+$, $A_i\leq I_i$, and $\sum_i A_i\leq1$, we bound the utility of $\pi$ as follows.
\begin{align}
u(\pi)
&= \mathbb{E}\!\left[\sum_i A_i \cdot v_i\right]
   - \sum_i \mathbb{E}[I_i \cdot c_i] \notag \\
&= \mathbb{E}\!\left[\sum_i A_i \cdot y_i + A_i \cdot o_i\right]
   - \sum_i \mathbb{E}[I_i \cdot c_i] \tag{by $v_i = y_i + o_i$} \\
&\leq
   \mathbb{E}\!\left[
      t \cdot \sum_i A_i
      + \sum_i A_i \cdot (y_i-t)^+
      + \sum_i A_i \cdot o_i
   \right]
   - \sum_i \mathbb{E}[I_i \cdot c_i] \notag \\
&\leq
   t
   + \sum_i \mathbb{E}\!\left[I_i \cdot (y_i-t)^+\right]
   + \sum_i \mathbb{E}[I_i \cdot o_i]
   - \sum_i \mathbb{E}[I_i \cdot c_i].\notag
\end{align}

For each box $i$, the decision to open it is made before its value is observed.
Thus, the event $\{I_i=1\}$ depends only on observations of other boxes and the strategy's internal randomness, so it is independent of $v_i$.
Both $y_i$ and $o_i$ are deterministic functions of $v_i$ (with $r_i$ fixed), so $I_i$ is also independent of $y_i$ and $o_i$. 
By the definition of reservation values, we have $\mathbb{E}[o_i] = c_i$, hence the term $\sum_i \mathbb{E}[I_i \cdot o_i]$ cancels with $\sum_i \mathbb{E}[I_i \cdot c_i]$.
Similarly, the definition of $w_i(t)$ gives $\mathbb{E}\!\left[I_i \cdot (y_i-t)^+\right] = \mathbb{E}\!\left[I_i \cdot w_i(t)\right]$.
Substituting these identities into the preceding bound yields
\[
u(\pi)\leq t+\mathbb{E}\!\left[\sum_i I_i \cdot w_i(t)\right] \leq t + W(t),
\]
where the last inequality holds because the opened boxes form a deadline-feasible set in every realization.
\end{proof}

Taking the supremum over the benchmark strategies and then minimizing over $t$ gives $\OPT\leq\min_{t\geq 0}\max_{S\in\D}\mathcal{B}_t(S)$.

To evaluate $t+W(t)$ for a fixed threshold $t$, we maximize the total weight of a deadline-feasible set using weights $w_i(t)$.
The following lemma characterizes feasibility and gives a dynamic program for this weighted selection problem.

\begin{lemma}[Feasibility test and weighted dynamic program]
Let $S=\{i_1,\ldots,i_k\}$ be ordered by nondecreasing deadlines. Then
$S\in\D$ if and only if $d_{i_j}\geq j$ for every $j=1,\ldots,k$. Given
nonnegative weights $a_i$, a set maximizing $\sum_{i\in S}a_i$ over
$S\in\D$ can be found with $O(n^2)$ state transitions.
\end{lemma}

\begin{proof}
If $d_{i_j}\geq j$ for every $j$, opening the boxes in deadline order is
legal. Conversely, in any legal schedule the $j$ boxes with the smallest
deadlines must all be opened within the first $d_{i_j}$ rounds. Since they
occupy distinct rounds, $d_{i_j}\geq j$.

Sort all boxes so that $d_1\leq\cdots\leq d_n$. Let $D(i,j)$ be the maximum
weight of a feasible set of exactly $j$ boxes chosen from the first $i$. Set
$D(0,0)=0$ and $D(0,j)=-\infty$ for $j>0$. Skipping box $i$ gives
$D(i-1,j)$. Taking it as the $j$th selected box is allowed only if $d_i\geq j$
and gives $D(i-1,j-1)+a_i$. Thus
\[
 D(i,j)=\max\{D(i-1,j),D(i-1,j-1)+a_i\}\quad\text{if }d_i\geq j.
\]
If $d_i<j$, only the skip term is allowed; when $j=0$, the take term is omitted. The optimum is
$\max_{0\leq j\leq n}D(n,j)$. There are $O(n^2)$ states and constant work per
transition.
\end{proof}

Applying the dynamic program with $a_i=w_i(t)$ computes $W(t)$ and a set maximizing $\mathcal{B}_t(S)$ with $O(n^2)$ state transitions.

\subsection{Immediate Commitment via a Balanced Threshold}
We next analyze a simple threshold policy with immediate commitment.
We then choose a balancing threshold and a feasible set to obtain the $1/2$ guarantee.

For a feasible set $S$ and a threshold $t\geq 0$, let
$S_t=\{i\in S:w_i(t)>0\}$. The policy $T^+(S,t)$ opens the boxes in $S_t$ in
nondecreasing deadline order. It immediately accepts the first box whose value
exceeds $t$ and stops, irrevocably rejecting every earlier box.
If no value exceeds $t$, it accepts the last box immediately upon opening it;
if $S_t$ is empty, it stops without opening or selecting a box.

\begin{lemma}[Threshold-surplus lower bound]
For every $S\in\D$ and $t\geq 0$,
\[
 u(T^+(S,t))\geq \min\!\left\{t,\sum_{i\in S}w_i(t)\right\}.
\]
\end{lemma}

\begin{proof}
Fix any $S\in\D$ and $t\geq 0$.
Remove from $S$ the boxes with zero surplus at threshold $t$, and order the remaining boxes by nondecreasing deadline as $i_1,\ldots,i_m$. 
Since $w_{i_j}(t)>0$, we have $r_{i_j}>t$.
Therefore, $y_{i_j}>t$ if and only if $v_{i_j}>t$.

For each position $j$, let $p_j=\Pr(y_{i_j}>t)$ denote the probability that box $i_j$ exceeds the threshold. 
Let $q_j=\prod_{\ell<j}(1-p_\ell)$ denote the probability that the policy reaches and opens box $i_j$.

If box $i_j$ is the first box whose value exceeds $t$, the policy stops after opening it. 
The deadline-feasible order ensures $d_{i_j}\geq j$, so the policy commits to this box by its deadline.
If no box exceeds $t$, the policy's reward is nonnegative, so we may omit that case when lower-bounding the expected reward.
The expected cost of opening box $i_j$ is $q_j \cdot c_{i_j}$. Summing the rewards from threshold crossings and subtracting the expected opening costs gives
\[
u\bigl(T^+(S,t)\bigr)
\geq
\sum_{j=1}^m q_j \cdot
\left(
\mathbb{E}\!\left[v_{i_j} \cdot \mathbf{1}_{\{v_{i_j}>t\}}\right]
-c_{i_j}
\right).
\]

The reservation-value identity lets us express each box's net contribution in terms of $w_{i_j}(t)$.
Since $r_{i_j}>t$, the overflow $o_{i_j}$ is zero whenever $v_{i_j}\leq t$, and $c_{i_j}=\mathbb{E}[o_{i_j}]$. 
Using $v_{i_j}=y_{i_j}+o_{i_j}$ and $y_{i_j}>t$ if and only if $v_{i_j}>t$, we get
\begin{align*}
&
\mathbb{E}\!\left[v_{i_j} \cdot \mathbf{1}_{\{v_{i_j}>t\}}\right]-c_{i_j} \\
= \, &
\mathbb{E}\!\left[y_{i_j} \cdot \mathbf{1}_{\{y_{i_j}>t\}}\right]
+\mathbb{E}\!\left[o_{i_j} \cdot \mathbf{1}_{\{v_{i_j}>t\}}\right]
-\mathbb{E}[o_{i_j}]\\
= \, &
\mathbb{E}\!\left[y_{i_j} \cdot \mathbf{1}_{\{y_{i_j}>t\}}\right] \\
= \, &
t\cdot p_j+\mathbb{E}\!\left[(y_{i_j}-t)^+\right] \\
= \, &
t\cdot p_j+w_{i_j}(t).
\end{align*}
The second equality uses that $o_{i_j}$ can be positive only when $v_{i_j}>t$; the third separates the threshold contribution from the excess above $t$.

Let $p=1-\prod_{j=1}^m(1-p_j)$ be the probability that at least one box exceeds $t$. 
Substituting the identity above gives
\[
u\bigl(T^+(S,t)\bigr)
\geq
t\,p+\sum_{j=1}^m q_j w_{i_j}(t).
\]
For every $j$, $q_j=\prod_{\ell<j}(1-p_\ell)\geq\prod_{\ell=1}^m(1-p_\ell)=1-p$. Therefore,
\[
u\bigl(T^+(S,t)\bigr)
\geq
t\,p+(1-p)\sum_{j=1}^m w_{i_j}(t)
\geq
\min\!\left\{t,\sum_{i\in S}w_i(t)\right\}.
\]
The middle expression is a convex combination of $t$ and $\sum_{i\in S}w_i(t)$.
Since the removed boxes have zero surplus, this proves the lemma.
\end{proof}

Combining the upper bound $t+W(t)$ on $\OPT$ with the threshold-surplus lower bound $\min\{t,W(t)\}$, we show that balancing the threshold and surplus terms yields the $1/2$ guarantee.
First, find a threshold $t^*$ satisfying $W(t^*)=t^*$.
Next, apply the dynamic program with weights $w_i(t^*)$ to find a feasible set $S^*$ maximizing their sum.
Finally, discard boxes with zero surplus and run $\mathrm{ALG}=T^+(S^*,t^*)$ in nondecreasing deadline order.
We state the guarantee below and discuss the computation of $t^*$ in its proof.

\begin{theorem}[$1/2$-approximation with immediate commitment]
\label{thm:half-approximation}
For Pandora's problem with deadlines, there exists a partially adaptive algorithm $\mathrm{ALG}$ with immediate commitment such that
\[
 u(\mathrm{ALG})\geq\frac{1}{2}\OPT,
\]
where $\OPT$ is defined among fully adaptive strategies with deferred-until-stopping commitment. 
\end{theorem}

\begin{proof}
We first verify that the construction is well defined.
Each function $w_i(t)=\E[\pos{y_i-t}]$ is continuous and nonincreasing, and converges to zero as $t\to\infty$.
As a maximum of finitely many sums of such functions, $W$ is also continuous and nonincreasing, with $W(t)\to 0$.
Hence $W(t)-t$ is continuous and strictly decreasing.
If $W(0)=0$, its unique zero is $t^*=0$.
The adaptive-surplus upper bound at $t=0$ then gives $\OPT\leq 0$, so stopping immediately has optimal utility zero and proves the theorem in this case.

Suppose now that $W(0)>0$.
Since $W(t)-t$ is positive at zero and tends to $-\infty$, it has a unique zero $t^*>0$.
A maximizing set $S^*$ exists because $\D$ is finite.
The threshold, selected set, and opening order are fixed before any values are observed, while the stopping time depends on the observations. Thus, $\mathrm{ALG}$ is partially adaptive.
The boxes retained by $T^+(S^*,t^*)$ form a subset of $S^*$, so their deadline order is feasible.
The algorithm makes an irrevocable acceptance or rejection decision immediately upon opening each box and stops upon acceptance.
Thus, whenever the algorithm commits to a box, it does so by that box's deadline.

By the threshold-surplus lower bound and the choice of $S^*$,
\[
 u(\mathrm{ALG})=u(T^+(S^*,t^*))
 \geq\min\{t^*,W(t^*)\}=t^*.
\]
The adaptive-surplus upper bound applies to every benchmark strategy, even when the box to which it commits has expired by the stopping time.
Taking the supremum over these strategies gives
\[
 \OPT\leq t^*+W(t^*)=2t^*.
\]
Combining the two bounds yields $u(\mathrm{ALG})\geq\frac12\OPT$, as claimed.

We next discuss how to compute $t^*$.
In general, we can first evaluate $W(0)$ and then apply binary search on $[0,W(0)]$.
Since $W$ is $n$-Lipschitz and $t^*\geq W(0)/(n+1)$,
$O(\log(n/\varepsilon))$ iterations suffice for any $0<\varepsilon<1/2$.
The resulting threshold gives an approximation ratio of $1/2-\varepsilon$, assuming exact evaluations of $W$.

For explicitly represented finite-support distributions, we can instead formulate the exact computation of $t^*$ as a linear program.
Let $y_{ik}$ be the capped values of box $i$ with probabilities $p_{ik}$, for $k=1,\ldots,m_i$, and let $M=\sum_i m_i$.
Let $J_i$ be the set of reachable cardinalities after the first $i$ boxes in the dynamic program. Introduce variables $t$ and $z_{ik}$, and state variables $L_{i,j}$ for $j\in J_i$. Write $a_i(z)=\sum_k p_{ik}z_{ik}$.
Using the same deadline order and transitions as in the dynamic program, consider
\[
\begin{aligned}
 \min\quad &t\\
 \text{s.t.}\quad
 &t\geq 0,\\
 &z_{ik}\geq y_{ik}-t,\quad z_{ik}\geq 0
     &&\quad i\in [n],\, k=1,\ldots,m_i,\\
 &L_{0,0}=0,\\
 &L_{i,j}\geq L_{i-1,j}
     &&\quad i\in [n],\, j\in J_{i-1},\\
 &L_{i,j}\geq L_{i-1,j-1}+a_i(z)
     &&\quad i\in [n],\, j-1\in J_{i-1},\, j\leq d_i,\\
 &L_{n,j}\leq t
     &&\quad j\in J_n.
\end{aligned}
\]
The optimal value of this LP is exactly $t^*$, and it has $O(M+n^2)$ variables and constraints.
For rational support values, probabilities, and costs encoded in binary, it can be solved exactly in polynomial time in the input length.
Running the dynamic program at $t^*$ then recovers $S^*$ and gives the exact $1/2$ guarantee.
\end{proof}

\begin{theorem}[Tightness of the $1/2$ guarantee]
\label{thm:immediate-commitment-tightness}
The factor $1/2$ is optimal among all immediate-commitment policies, even fully adaptive ones, against the best partially adaptive strategy with deferred-until-stopping commitment.
\end{theorem}

\begin{proof}
Tightness follows by simulating the classical tight example for prophet inequalities.
Consider two zero-cost boxes with $d_1=1$, $v_1=1$, $d_2=2$, and $v_2=1/\varepsilon$ with probability $\varepsilon$ and $0$ otherwise, where $0<\varepsilon<1$.
A partially adaptive strategy with deferred-until-stopping commitment opens box 1 and then box 2 and chooses the larger value, obtaining expected utility $2-\varepsilon$.
Any immediate-commitment policy has expected utility at most $1$: opening box 2 first precludes opening box 1 later, while opening box 1 first offers a choice between accepting $1$ and continuing to box 2, whose expectation is also $1$.
Thus, the ratio is at most $1/(2-\varepsilon)\to1/2$, proving tightness even for fully adaptive immediate-commitment policies.
\end{proof}

\section{Hardness of Partial Adaptivity}
\label{sec:eg}

Section \ref{sec:0.5} shows that a partially adaptive policy can obtain a $1/2$-approximation even under immediate commitment. 
We now show that partial adaptivity itself has an inherent limitation: even under deferred-until-stopping commitment, its worst-case approximation ratio relative to full adaptivity is bounded away from one.

Let $\mathrm{OPT}_{\mathrm{PA}}(\mathcal I)$ and $\mathrm{OPT}_{\mathrm{FA}}(\mathcal I)$ denote, respectively, the supremum expected utilities of partially adaptive and fully adaptive policies on instance $\mathcal I$, both under deferred-until-stopping commitment.

\begin{theorem}[Limitation of partial adaptivity]
\label{thm:partial-adaptivity-limit}
There exists an instance $\mathcal I$ of Pandora’s Problem with Deadlines such that
\[
\mathrm{OPT}_{\mathrm{PA}}(\mathcal I)
<
0.778 \cdot \mathrm{OPT}_{\mathrm{FA}}(\mathcal I).
\]
Thus, no partially adaptive policy, even with unlimited computational power, can guarantee a worst-case approximation ratio greater than $0.778$ against the fully adaptive optimum.
\end{theorem}

\begin{proof}
Fix a sufficiently large integer $k$. 
Construct an instance with at least $k$ boxes of each of two types. 
Every box has deadline $k$ and opening cost zero, so at most $k$ boxes can be opened.
All box values are mutually independent. A type-1 box has value $X_1$, and a type-2 box has value $X_2$, where
\[
X_1 =
\begin{cases}
1 & \text{with probability } p_1,\\
0 & \text{otherwise,}
\end{cases}
\qquad
X_2 =
\begin{cases}
k & \text{with probability } p_2,\\
0 & \text{otherwise.}
\end{cases}
\]
Choose fixed constants $0<\lambda_1\leq \lambda_2$ and set
\[
p_1=\frac{1}{\lambda_1 k},
\qquad
p_2=\frac{1}{\lambda_2 k^2}.
\]
For sufficiently large $k$, both are valid probabilities. 
Since opening is free and all opened boxes remain available until stopping, the utility is the largest value among the opened boxes.

\paragraph{Optimal partially adaptive strategy.}
A partially adaptive strategy fixes its opening order in advance. 
Because opening another box is free and cannot decrease the maximum observed value, an optimal such strategy opens exactly $k$ boxes. 
If it opens $\ell$ type-1 boxes and $k-\ell$ type-2 boxes, its expected utility is
\[
\bigl[1-(1-p_1)^\ell\bigr](1-p_2)^{k-\ell}
+
k\bigl[1-(1-p_2)^{k-\ell}\bigr].
\]
The first term is the probability of seeing a type-1 value of $1$ and no type-2 value of $k$; the second is the expected contribution when at least one type-2 box yields $k$. 
Thus,
\[
\begin{aligned}
\mathrm{OPT}_{\mathrm{PA}}^{(k)}
= \max_{\ell\in\{0,\ldots,k\}} \Biggl\{
&\left[1-\left(1-\frac{1}{\lambda_1 k}\right)^\ell\right]
 \left(1-\frac{1}{\lambda_2 k^2}\right)^{k-\ell}
+ k\left[1-\left(1-\frac{1}{\lambda_2 k^2}\right)^{k-\ell}\right]
\Biggr\}.
\end{aligned}
\]
Writing $\alpha=\ell/k$, the expression inside the maximum converges uniformly over $\ell$ to
\[
g(\alpha)=1-e^{-\alpha/\lambda_1}+\frac{1-\alpha}{\lambda_2}.
\]
Consequently, when $k\to\infty$,
\[
\mathrm{OPT}_{\mathrm{PA}}^{(k)}
=
\max_{\alpha\in[0,1]} g(\alpha)+o(1).
\]
For our choice of constants, $g$ is maximized at
$\alpha^*=\lambda_1\ln(\lambda_2/\lambda_1)$, which lies in $[0,1]$. Hence,
\[
\mathrm{OPT}_{\mathrm{PA}}^{(k)}
=
1-\frac{\lambda_1}{\lambda_2}
+\frac{1-\lambda_1\ln(\lambda_2/\lambda_1)}{\lambda_2}
+o(1).
\]

\paragraph{A fully adaptive strategy.}
Consider the following fully adaptive policy. 
It opens type-1 boxes until it observes value $1$ or uses all $k$ rounds. 
If it observes $1$ after $t$ type-1 failures, it opens type-2 boxes in each of the remaining $k-t-1$ rounds and, at stopping, selects the highest-valued opened box. 
This is a feasible fully adaptive policy; we do not require it to be optimal\footnote{In fact, this fully adaptive policy is optimal for the constructed instance. 
This can be verified by backward induction on the number of rounds remaining, comparing the expected continuation values of opening each type of box after every possible history.}.

Its expected utility is
\[
U_{\mathrm{FA}}^{(k)}
=
1-(1-p_1)^k
+
(k-1)\sum_{t=0}^{k-1}
(1-p_1)^t p_1
\left[1-(1-p_2)^{k-t-1}\right].
\]
Here $(1-p_1)^t p_1$ is the probability that the first type-1 success occurs after exactly $t$ failures. 
Conditional on that event, the bracketed term is the probability that one of the remaining type-2 boxes yields value $k$, upgrading the reward from $1$ to $k$.

As $k\to\infty$, the sum is a Riemann sum, and
\[
U_{\mathrm{FA}}^{(k)}
=
1-e^{-1/\lambda_1}
+
\frac{1-\lambda_1(1-e^{-1/\lambda_1})}{\lambda_2}
+o(1).
\]

\paragraph{Approximation ratio.}

Since the above fully adaptive policy is feasible,
$\mathrm{OPT}_{\mathrm{FA}}^{(k)}\geq U_{\mathrm{FA}}^{(k)}$. 
Therefore,
\[
\frac{\mathrm{OPT}_{\mathrm{PA}}^{(k)}}
     {\mathrm{OPT}_{\mathrm{FA}}^{(k)}}
\leq
\frac{\mathrm{OPT}_{\mathrm{PA}}^{(k)}}{U_{\mathrm{FA}}^{(k)}}.
\]

Set $\lambda_1=0.513$ and $\lambda_2=1.541$. The ratio of
the limiting utilities is
\[
R=
\frac{
1-\lambda_1/\lambda_2+
\bigl(1-\lambda_1\ln(\lambda_2/\lambda_1)\bigr)/\lambda_2
}{
1-e^{-1/\lambda_1}+
\bigl(1-\lambda_1(1-e^{-1/\lambda_1})\bigr)/\lambda_2
}.
\]
For these rational parameter values, bounding the logarithm and exponential
(for example, by Taylor expansions with remainder) gives
\[
R<0.77792<0.778.
\]
Since
\[
\frac{\mathrm{OPT}_{\mathrm{PA}}^{(k)}}
     {\mathrm{OPT}_{\mathrm{FA}}^{(k)}}
\leq
\frac{\mathrm{OPT}_{\mathrm{PA}}^{(k)}}{U_{\mathrm{FA}}^{(k)}}
\longrightarrow R,
\]
the ratio is strictly less than $0.778$ for all sufficiently large finite
$k$. This gives the required instance and proves the theorem.
\end{proof}

\section{Conclusion}
We presented a simple partially adaptive algorithm that commits immediately and achieves a $1/2$-approximation against the fully adaptive optimum with deferred-until-stopping commitment.
Its analysis directly combines Weitzman's reservation-value idea with the threshold--surplus decomposition underlying prophet inequalities.
We also proved that the $1/2$ guarantee is tight for immediate-commitment policies against this benchmark, and that partially adaptive policies cannot guarantee a ratio above $0.778$ even under deferred-until-stopping commitment.

Although immediate commitment and partial adaptivity can be advantageous in practice, determining the tight approximation ratio for more general efficient algorithms remains of theoretical interest.
One direction is to exploit deferred-to-deadline commitment, which allows a policy to gather additional observations while retaining the option to select earlier boxes that have not expired.
Neither our algorithm nor the prior algorithms \citep{DBLP:conf/aaai/BergerEFF24,DBLP:journals/ai/AmanatidisBEFFRT25} exploit this flexibility.
Can it improve the approximation ratio beyond $1/2$?

A second direction is to exploit full adaptivity in the opening order.
In the classical Pandora's problem, Weitzman's partially adaptive rule is optimal even against fully adaptive policies, whereas our separation shows that deadlines make partial adaptivity strictly weaker.
Under a common deferred commitment rule, can efficient fully adaptive algorithms achieve better approximation guarantees than partially adaptive algorithms?

\section*{Declaration of AI Usage}
All the theoretical results in this paper, including the algorithm, its analysis, and the proofs, were developed without the assistance of AI tools. 
During the preparation of the manuscript, we used ChatGPT Astra to proofread and polish the writing. 
We carefully reviewed its presentation and take full responsibility for all the content.
\bibliographystyle{plainnat}
\bibliography{pandora}

@inproceedings{BeyhaghiK:EC:2019,
  title={Pandora's problem with nonobligatory inspection},
  author={Beyhaghi, Hedyeh and Kleinberg, Robert},
  booktitle={Conference on Economics and Computation, {EC}},
  pages={131--132},
  year={2019}
}

@article{DBLP:journals/sigecom/BeyhaghiC23,
  author       = {Hedyeh Beyhaghi and
                  Linda Cai},
  title        = {Recent Developments in Pandora's Box Problem: Variants and Applications},
  journal      = {SIGecom Exch.},
  volume       = {21},
  number       = {1},
  pages        = {20--34},
  year         = {2023}
}

@inproceedings{DBLP:conf/approx/0001GMT23,
  author       = {Shuchi Chawla and
                  Evangelia Gergatsouli and
                  Jeremy McMahan and
                  Christos Tzamos},
  title        = {Approximating Pandora's Box with Correlations},
  booktitle    = {{APPROX/RANDOM}},
  volume       = {275},
  pages        = {26:1--26:24},
  year         = {2023}
}

@inproceedings{GergatsouliT:NeurIPS:2023,
  author       = {Evangelia Gergatsouli and
                  Christos Tzamos},
  title        = {Weitzman's rule for {P}andora's box with correlations},
  booktitle    = {  Neural Information Processing Systems, NeurIPS},
  year         = {2023}
}

@inproceedings{ChawlaGTTZ:FOCS:2020,
  author       = {Shuchi Chawla and
                  Evangelia Gergatsouli and
                  Yifeng Teng and
                  Christos Tzamos and
                  Ruimin Zhang},
  title        = {Pandora's Box with Correlations: Learning and Approximation},
  booktitle    = {Foundations of Computer Science, {FOCS}},
  pages        = {1214--1225},
  year         = {2020}
}

@article{Weitzman/1979/original,
 author = {Martin L. Weitzman},
 journal = {Econometrica},
 number = {3},
 pages = {641--654},
 publisher = {[Wiley, Econometric Society]},
 title = {Optimal Search for the Best Alternative},
 urldate = {2024-06-28},
 volume = {47},
 year = {1979}
}

@inproceedings{DBLP:conf/soda/Singla18,
  author       = {Sahil Singla},
  title        = {The Price of Information in Combinatorial Optimization},
  booktitle    = {Symposium on discrete Algorithms, {SODA}},
  pages        = {2523--2532},
  year         = {2018}
}

@inproceedings{DBLP:conf/icml/GergatsouliT22,
  author       = {Evangelia Gergatsouli and
                  Christos Tzamos},
  title        = {Online Learning for Min Sum Set Cover and Pandora's Box},
  booktitle    = {{ICML}},
  series       = {Proceedings of Machine Learning Research},
  volume       = {162},
  pages        = {7382--7403},
  publisher    = {{PMLR}},
  year         = {2022}
}

@inproceedings{DBLP:conf/sigecom/KleinbergWW16,
  author       = {Robert D. Kleinberg and
                  Bo Waggoner and
                  E. Glen Weyl},
  title        = {Descending Price Optimally Coordinates Search},
  booktitle    = {Conference on Economics and Computation, {EC}},
  pages        = {23--24},
  year         = {2016}
}

@article{BoodaghiansFLL:Math.Oper:2023,
  author       = {Shant Boodaghians and
                  Federico Fusco and
                  Philip Lazos and
                  Stefano Leonardi},
  title        = {Pandora's Box Problem with Order Constraints},
  journal      = {Math. Oper. Res.},
  volume       = {48},
  number       = {1},
  pages        = {498--519},
  year         = {2023}
}

@inproceedings{BergerEFF:EC:2023,
  author       = {Ben Berger and
                  Tomer Ezra and
                  Michal Feldman and
                  Federico Fusco},
  title        = {Pandora's Problem with Combinatorial Cost},
  booktitle    = {{EC}},
  pages        = {273--292},
  publisher    = {{ACM}},
  year         = {2023}
}

@article{DBLP:journals/corr/abs-0804-1724,
  author       = {Sudipto Guha and
                  Kamesh Munagala and
                  Saswati Sarkar},
  title        = {Information Acquisition and Exploitation in Multichannel Wireless
                  Networks},
  journal      = {CoRR},
  volume       = {abs/0804.1724},
  year         = {2008}
}

@inproceedings{DBLP:conf/stoc/FuLL23,
  author       = {Hu Fu and
                  Jiawei Li and
                  Daogao Liu},
  title        = {Pandora Box Problem with Nonobligatory Inspection: Hardness and Approximation
                  Scheme},
  booktitle    = {Sympoium on Theory of Computing, {STOC}},
  pages        = {789--802},
  year         = {2023}
}

@inproceedings{DBLP:conf/focs/BansalHZ25,
  author       = {Nikhil Bansal and
                  Zhiyi Huang and
                  Zixuan Zhu},
  title        = {Optimal 4-Approximation for the Correlated Pandora's Problem},
  booktitle    = {{FOCS}},
  pages        = {1813--1828},
  publisher    = {{IEEE}},
  year         = {2025}
}

@inproceedings{DBLP:conf/stoc/BeyhaghiC23,
  author       = {Hedyeh Beyhaghi and
                  Linda Cai},
  title        = {Pandora's Problem with Nonobligatory Inspection: Optimal Structure
                  and a {PTAS}},
  booktitle    = {Symposium on Theory of Computing, {STOC}},
  pages        = {803--816},
  year         = {2023}
}

@article{DBLP:journals/jet/Doval18,
  author       = {Laura Doval},
  title        = {Whether or not to open Pandora's box},
  journal      = {J. Econ. Theory},
  volume       = {175},
  pages        = {127--158},
  year         = {2018}
}

@article{aouad2020pandora,
  author  = {Aouad, Ali and Ji, Jingwei and Shaposhnik, Yaron},
  title   = {The Pandora's Box Problem with Sequential Inspections},
  journal = {Operations Research},
  volume  = {74},
  number  = {4},
  pages   = {2179--2195},
  year    = {2026},
  doi     = {10.1287/opre.2024.0733}
}

@inproceedings{EsfandiariHLM:AAAI:2019,
  author       = {Hossein Esfandiari and
                  Mohammad Taghi Hajiaghayi and
                  Brendan Lucier and
                  Michael Mitzenmacher},
  title        = {Online Pandora's Boxes and Bandits},
  booktitle    = {{AAAI}},
  pages        = {1885--1892},
  publisher    = {{AAAI} Press},
  year         = {2019}
}

@inproceedings{DBLP:conf/aaai/BergerEFF24,
  author       = {Ben Berger and
                  Tomer Ezra and
                  Michal Feldman and
                  Federico Fusco},
  title        = {Pandora's Problem with Deadlines},
  booktitle    = {{AAAI}},
  pages        = {20337--20343},
  publisher    = {{AAAI} Press},
  year         = {2024}
}

@article{ZivL:corr:2024,
  author       = {Ziv Scully and
                  Laura Doval},
  title        = {Local hedging approximately solves Pandora's box problems with
                  nonobligatory inspection},
  journal      = {CoRR},
  volume       = {abs/2410.19011},
  year         = {2024}
}

@inproceedings{FuLX:ICALP:2018,
  title={A {PTAS} for a class of stochastic dynamic programs},
  author={Fu, Hao and Li, Jian and Xu, Pan},
  booktitle={ International Colloquium on Automata, Languages, and Programming, {ICALP}},
  pages={56},
  year={2018}
}

@inproceedings{DBLP:conf/sigecom/Segev021,
  author       = {Danny Segev and
                  Sahil Singla},
  title        = {Efficient Approximation Schemes for Stochastic Probing and Prophet
                  Problems},
  booktitle    = {Conference on Economics and Computation, {EC}},
  pages        = {793--794},
  year         = {2021}
}

@inproceedings{DBLP:conf/wine/AmanatidisFRT24,
  author       = {Georgios Amanatidis and
                  Federico Fusco and
                  Rebecca Reiffenh{\"{a}}user and
                  Artem Tsikiridis},
  title        = {Pandora's Box Problem Over Time},
  booktitle    = {{WINE}},
  series       = {Lecture Notes in Computer Science},
  volume       = {15534},
  pages        = {494--512},
  publisher    = {Springer},
  year         = {2024}
}

@article{KrengelSucheston1977,
  author  = {Krengel, Ulrich and Sucheston, Louis},
  title   = {Semiamarts and Finite Values},
  journal = {Bulletin of the American Mathematical Society},
  volume  = {83},
  number  = {4},
  pages   = {745--747},
  year    = {1977}
}

@article{KrengelSucheston1978,
  author  = {Krengel, Ulrich and Sucheston, Louis},
  title   = {On semiamarts, amarts, and processes with finite value},
  journal = {Adv. in Probability},
  volume  = {4},
  pages   = {197--266},
  year    = {1978}
}

@article{SamuelCahn1984,
  author  = {Samuel-Cahn, Ester},
  title   = {Comparison of Threshold Stop Rules and Maximum for Independent Nonnegative Random Variables},
  journal = {The Annals of Probability},
  volume  = {12},
  number  = {4},
  pages   = {1213--1216},
  year    = {1984}
}

@article{DBLP:journals/ai/AmanatidisBEFFRT25,
  author       = {Georgios Amanatidis and
                  Ben Berger and
                  Tomer Ezra and
                  Michal Feldman and
                  Federico Fusco and
                  Rebecca Reiffenh{\"{a}}user and
                  Artem Tsikiridis},
  title        = {Pandora's box problem with time constraints},
  journal      = {Artif. Intell.},
  volume       = {349},
  pages        = {104426},
  year         = {2025}
}


\end{document}